\documentclass[11pt]{amsart}
\usepackage{amssymb, amsmath}
\usepackage{graphicx}
\pdfoutput=1

\usepackage{subfig, tikz}
\usetikzlibrary{decorations.pathmorphing}
\usepackage{graphics}
\usepackage{hyperref}

\newcommand{\td}{\text{d}}

\usepackage{bm}
\usepackage{enumerate}
\theoremstyle{plain}
\newtheorem{theorem}{Theorem}[section]

\newtheorem{cor}[theorem]{Corollary}

\newtheorem{proposition}[theorem]{Proposition}
\theoremstyle{definition}
\newtheorem{remark}[theorem]{Remark}

\usepackage[%
	style=phys, 
	articletitle=false, 
	biblabel=brackets, 
	eprint=true,
	chaptertitle=false, 
	pageranges=false
]{biblatex}
\bibliography{ref}

\def\be{\begin{equation}}
\def\ee{\end{equation}}
\def\bea{\begin{eqnarray}}
\def\eea{\end{eqnarray}}

\newcounter{mnotecount}[section]

\renewcommand{\themnotecount}{\thesection.\arabic{mnotecount}}

\newcommand{\mnote}[1]
{\protect{\stepcounter{mnotecount}}$^{\mbox{\footnotesize
$
\bullet$\themnotecount}}$ \marginpar{
\raggedright\tiny\em
$\!\!\!\!\!\!\,\bullet$\themnotecount: #1} }

\begin{document}
\title[]{Rigidity of the extremal Kerr-Newman horizon}
\author{Alex Colling}
\address{Department of Applied Mathematics and Theoretical Physics\\ 
University of Cambridge\\ Wilberforce Road, Cambridge CB3 0WA, UK.}
\email{aec200@cam.ac.uk}
\author{David Katona}
\address{
School of Mathematics and Maxwell Institute for Mathematical Sciences\\ University of Edinburgh\\
King's Buildings, Edinburgh, EH9 3FD, UK.}
\email{d.katona@sms.ed.ac.uk}
\author{James Lucietti}
\address{
School of Mathematics and Maxwell Institute for Mathematical Sciences\\ University of Edinburgh\\
King's Buildings, Edinburgh, EH9 3FD, UK.}
\email{j.lucietti@ed.ac.uk}


\begin{abstract}
We prove that the intrinsic geometry of compact cross-sections of an extremal horizon in four-dimensional Einstein-Maxwell theory must admit a Killing vector field or is static. This implies that any such horizon must be an extremal Kerr-Newman horizon and completes the classification of the associated near-horizon geometries.  The same results  hold with a cosmological constant. 
\end{abstract}

\maketitle

\section{Introduction}

The no-hair theorem states that, under certain global conditions, the Kerr-Newman black hole is the only equilibrium black hole spacetime in Einstein-Maxwell theory~\cite{Chrusciel:2008js}. This celebrated result rests on several remarkable theorems that constrain the topology,  symmetry and geometry of such black hole spacetimes.  In particular, Hawking's rigidity theorem establishes that the event horizon of any stationary (analytic) spacetime must  be a Killing horizon, and furthermore, if the black hole is rotating the spacetime must be axially symmetric~\cite{Hawking:1971vc, Hollands:2006rj, Moncrief:2008mr}. This rigidity theorem was originally proven for non-extremal black holes and has been generalised to include extremal black holes~\cite{Chrusciel:1996bm,Hollands:2008wn}. Interestingly, although the proof exploits extrinsic properties of the horizon, it implies that the intrinsic geometry of the horizon must be axially symmetric.  

It is well known that the Einstein equations restricted to a Killing horizon imply that the intrinsic and extrinsic geometries decouple precisely if the horizon is extremal. The intrinsic geometry of an extremal horizon in an $(n+2)$-dimensional spacetime is described by a Riemannian metric and a vector field on an $n$-dimensional cross-section of the horizon that obey a quasi-Einstein equation. This quasi-Einstein structure is equivalent to an $(n+2)$-dimensional associated spacetime that satisfies the Einstein equations, known as the  near-horizon geometry, which can be obtained as a near-horizon scaling limit of the parent spacetime~\cite{Kunduri:2013gce}. This structure is also equivalent to the intrinsic geometry of an extremal isolated horizon~\cite{Lewandowski:2002ua}.

The geometry induced on an extremal horizon in Einstein-Maxwell theory is described by an $n$-dimensional Riemannian manifold $(M, g)$, a vector field $X\in \mathfrak{X}(M)$, a function $\psi$ and a closed 2-form $B\in \Omega^2(M)$ satisfying the system of equations~\cite{Kunduri:2013gce, Li:2018knr}
\begin{align}
	R_{ab} = \frac{1}{2}X_a X_b - \nabla_{(a}X_{b)} + \lambda g_{ab} + P_{ab}\;, \label{eq_NHG}\\
	(\nabla^a-X^a)B_{ab} = -(\nabla_b-X_b)\psi\;, \label{eq_NH_Maxwell}
\end{align}
where  $\nabla$ is the Levi-Civita connection and $R_{ab}$ is the Ricci tensor of $(M, g)$,  $\lambda \in \mathbb{R}$ is the cosmological constant, and the `source term'
\begin{equation}
	P_{ab}:= 2B_{ac}B_b{}^c + \frac{1}{n}g_{ab}\left(2\psi^2 - B_{cd}B^{cd}\right)\;.  \label{eq_source}
\end{equation}
We will refer to $(M, g)$ together with the data $(X, \psi, B)$ satisfying the above equations as a {\it quasi-Einstein-Maxwell structure} on $M$. In the original $(n+2)$-dimensional spacetime,  the manifold $M$ corresponds to a cross-section of the horizon, $g$ is the induced Riemannian metric on $M$, and the  function $\psi$ and the 2-form $B$ correspond to the electric potential and the Maxwell field induced on $M$, respectively. However, the quasi-Einstein-Maxwell structure may be studied independently of any spacetime interpretation, which is the perspective we take in this paper.

Numerous rigidity and classification results are known for  extremal horizons with a compact cross-section $M$~\cite{Kunduri:2013gce}. For instance, vacuum solutions (so $\psi= B=0$) with $X^\flat$ closed, where $X^\flat$ denotes the $g$-dual of the vector field $X$, must be trivial if $n=2$, or either trivial or locally the product of a circle and an Einstein space if $n > 2$, $\lambda \leq 0$ ~\cite{Chrusciel:2005pa, Bahuaud:2022iao, Wylie:2023pwf,kaminski_extreme_2024} (here trivial means $X$ vanishes identically). Therefore, the classification of static near-horizon geometries, which correspond to $X^\flat$ closed, is largely complete in the vacuum theory. Recently, it has been shown that for vacuum extremal horizons with $X$ non-gradient, the corresponding quasi-Einstein structure implies the existence of a Killing vector field on compact $(M,g)$ for all $n\geq 2$ and $\lambda \in \mathbb{R}$~\cite{Dunajski:2023xrd}. This is an intrinsic version of Hawking's rigidity theorem for extremal horizons. Combining with previous uniqueness theorems for axially symmetric extremal horizons~\cite{Lewandowski:2002ua, Kunduri:2008rs, Hajicek:1974oua}, this implies that the extremal Kerr horizon  is the unique solution on $M=S^2$ (including a possible cosmological constant). The purpose of this paper is to generalise this intrinsic rigidity theorem to Einstein-Maxwell theory.

We will  focus on the $n=2$ dimensional case, corresponding to extremal horizons in four-dimensional Einstein-Maxwell theory. In particular, we will show that any non-gradient quasi-Einstein-Maxwell structure on a compact surface must admit a Killing vector field and therefore,  as we explain below,  complete their classification.  Our main result is the following.
\begin{theorem}\label{main-theorem}
Let $(M, g)$ be a $2$-dimensional compact (without boundary), oriented, Riemannian manifold admitting a quasi-Einstein-Maxwell structure with a non-gradient vector field $X$. Then, there exists a smooth positive function $\Gamma$ such that $K=\Gamma X+ \nabla \Gamma$ is a Killing vector field. Furthermore, the vector field $X$ and the Maxwell data $(\psi, B)$ are invariant under $K$.
\end{theorem}

The proof proceeds in an identical manner to the vacuum case~\cite{Dunajski:2023xrd}. The main tool is a remarkable divergence identity that follows from the quasi-Einstein equations.  It turns out that an analogous identity can be established for the 2-dimensional quasi-Einstein-Maxwell equations.   Then, upon choosing $\Gamma$ to be the principal eigenfunction of a certain elliptic operator (which ensures that $K$ is divergence-free), the identity implies that the sum of the $g$-norm of the Lie derivative of the metric $\mathcal{L}_K g$ and of the $g$-norm of a quantity involving the Maxwell field data  is a total divergence. Therefore, by integrating this identity over compact $M$ and using Stokes' theorem, this implies that $K$ satisfies Killing's equation and constrains the Maxwell data  in such a way that one can deduce the rest of the data is also invariant under $K$.

Our theorem has a number of corollaries which allow us to complete the classification of quasi-Einstein-Maxwell structures on surfaces.  Firstly, the existence of a Killing field as in Theorem \ref{main-theorem}, together with the fact that $X$ must have a zero (which was previously shown in the vacuum case~\cite{Jezierski:2009zz,Jezierski:2012lzy, Chrusciel:2017vie, Bahuaud:2023jyq}), implies that $M=S^2$ (see Corollary \ref{cor1}). In particular, this rules out non-gradient quasi-Einstein-Maxwell structures on higher genus surfaces, a fact that had  previously only been proven for the vacuum theory~\cite{Dobkowski-Rylko:2018nti}\footnote{Since the release of this paper solutions on higher genus surfaces have been ruled out for a more general class of equations in~\cite{kaminski_extreme_2024}.}. Secondly, it has been known for some time that under the assumption of an axial Killing field that preserves the vector field $X$ and Maxwell field data, the only non-gradient solution on $M=S^2$ corresponds to an extremal Kerr-Newman horizon, possibly with a cosmological constant~\cite{Lewandowski:2002ua, Kunduri:2008rs, Hajicek:1974oua}. Therefore, in combination with Theorem \ref{main-theorem} we deduce the following.
\begin{theorem}
\label{main-theorem-cor}
	Any quasi-Einstein-Maxwell structure as in Theorem \ref{main-theorem} is an extremal Kerr-Newman horizon, possibly with a cosmological constant.
\end{theorem}

The above result is complementary to the classification of static near-horizon geometries, which as mentioned above correspond to $X^\flat$ closed.  In this case it has been shown that there are no non-trivial near-horizon geometries, that is, any quasi-Einstein-Maxwell structure on a surface must be trivial in the sense that $X=0$ and $(M,g)$ is Einstein~\cite{Chrusciel:2006pc, Kunduri:2008tk}.\footnote{The proof for $\lambda>0$ was not given explicitly in~\cite{Kunduri:2008tk}, but it can be generalised to cover this case~\cite{kaminski_extreme_2024,katona_classification_2024}.}  Taken with our Theorem \ref{main-theorem-cor}, this completes the classification of near-horizon geometries with compact cross-sections $M$ in four-dimensional Einstein-Maxwell theory including an arbitrary cosmological constant.

Generalising Theorem \ref{main-theorem} to higher-dimensional Einstein-Maxwell theory (possibly coupled to a Chern-Simons term) is an interesting problem, which will be the subject of future work. However, we note that even the classification of $n=3$ quasi-Einstein-Maxwell(-Chern-Simons) structures with two commuting axial Killing fields remains an open problem and multiple solutions are known in this case~\cite{Kunduri:2013gce}.\\

\noindent {\bf Acknowledgements.} AC acknowledges the support of the Cambridge International Scholarship. DK is supported by a School of Mathematics Studentship from the University of Edinburgh. We would like to thank Maciej Dunajski for useful discussions.\\

\section{Rigidity of extremal horizons in Einstein-Maxwell theory}

In this paper we will consider extremal horizons in four-dimensional Einstein-Maxwell theory. Therefore, we will assume $(M, g)$ is a $2$-dimensional (oriented) Riemannian manifold that satisfies the quasi-Einstein-Maxwell equations \eqref{eq_NHG}, \eqref{eq_NH_Maxwell}, \eqref{eq_source}.  Thus, we may define a function $\beta:= \star B$ where $\star$ is the Hodge dual operator on $(M, g)$. Then, the source term \eqref{eq_source} simplifies to
\begin{equation}
	P_{ab} = (\psi^2 + \beta^2)g_{ab}\;,  \label{eq_P_2d}
\end{equation}
and the Maxwell-equation (\ref{eq_NH_Maxwell}) becomes
\begin{equation}
	\star(\td\beta - \beta X^\flat) = \td\psi -X^\flat \psi\; .  \label{eq_NH_Maxwell_2}
\end{equation}
Adding (\ref{eq_NH_Maxwell_2}) multiplied  by $\psi$ to the Hodge-dual  of (\ref{eq_NH_Maxwell_2}) multiplied by $\beta$, yields
\begin{equation}
	X^\flat (\psi^2 + \beta^2) =  \star(\beta\td \psi-\psi\td\beta) +\frac{1}{2}\td (\psi^2 + \beta^2)\;. \label{eq_Maxwell_X}
\end{equation}
Thus, if we define the continuous function 
\begin{equation}
		\rho:= (\psi^2+\beta^2)^{1/2}  \; ,\label{eq_rhodef}
	\end{equation} 
	we deduce that on the open set $\widetilde M \subseteq M$  where $\rho>0$ we can write
	\begin{equation}
	X^\flat = -\star  \Phi + \td\log\rho  \; , \label{eq_Maxwell_X_2}
	\end{equation}
	where $\Phi$ is a  closed 1-form  on $\widetilde M$ defined by
	\begin{equation}
	\Phi := \frac{\psi\td\beta-\beta\td \psi}{\psi^2+\beta^2}  \; .\label{eq_Phidef}
	\end{equation}
	In fact, on $\widetilde M$ the Maxwell equation \eqref{eq_NH_Maxwell_2} is equivalent to \eqref{eq_Maxwell_X_2} since the operator $\psi+\beta \star$ acting on 1-forms is invertible.
If $\widetilde M$ is simply connected, then by the Poincar\'e lemma $\Phi = \td \phi$ for some smooth function $\phi$ on $\widetilde M$, and we can invert \eqref{eq_rhodef}, \eqref{eq_Phidef} to obtain
	\begin{equation}
	\psi = \rho \cos\phi\;, \qquad \beta = \rho\sin\phi\; , \label{eq_phidef}
\end{equation}
where we have fixed an additive integration constant for $\phi$.

We are now ready to establish our first result.
\begin{proposition}  \label{prop1}
Consider a quasi-Einstein-Maxwell structure on a compact (no boundary), oriented,  surface $M$. If $\rho$ does not vanish identically on $M$, then $\rho$ is nowhere vanishing on $M$. In this case, $\Phi$ defined by \eqref{eq_Phidef} is a globally defined closed 1-form on $M$ and \eqref{eq_Maxwell_X_2} holds everywhere on $M$.
\end{proposition}
\begin{proof}
	Let $\widetilde M \subseteq M$ be the open set on which $\rho>0$. By taking $\td \star$ of \eqref{eq_Maxwell_X_2}  we deduce that $\td \star \td \log \rho = \td \star X^\flat$ on $\widetilde M$. On the other hand, by standard results the Poisson equation $\td \star \td f=  \td \star X^\flat$ on a compact $(M,g)$ (with a fixed $X$) must admit a unique solution $f \in C^\infty(M)$ (up to an additive constant). Therefore, $h :=\log \rho - f$ is a harmonic function on $\widetilde M$. 
  
	Now assume for contradiction that $\partial \widetilde M \neq \emptyset$.  For any $p\in \partial \widetilde M$ we have  $\lim_{x\to p} \rho(x) = 0$ and hence $\lim_{x\to p} h(x) = -\infty$, which contradicts the maximum principle for harmonic functions. This can be seen as follows. For $\epsilon > 0$ let $M_\epsilon:=\{x\in M: \rho (x) \ge \epsilon\} \subset \widetilde M$. $M_\epsilon$ is compact, since it is closed (by the continuity of $\rho$), and it is a subset of a compact manifold $M$. Let us choose a sufficiently small $\epsilon>0$ such that $M_\epsilon$ is non-empty. Then, by the maximum principle, $h$ takes its maximum $h(q)$ at some $q\in\partial M_\epsilon$ (which is non-empty by our assumption that $\rho$ has a zero). Now, consider $M_\delta$ for some
	\begin{equation}
		0<\delta<\epsilon \frac{e^{\min_{M}f}}{ e^{\max_{M} f}}\leq \epsilon \;,
	\end{equation}
	where the maximum and minimum of $f$ must exist by compactness of $M$ (again, $M_\delta$ and its boundary are non-empty for any such $\delta$, by the assumption that $\rho$ has a zero on $M$). By  the maximum principle the maximum $h(q')$ over $M_\delta$ occurs at some boundary point $q'\in\partial M_\delta$. Then, since $M_\delta \supset M_\epsilon$, we must have $h(q')\ge h(q)$. On the other hand, by construction we have
	\begin{equation}
		h(q') = \log\delta -f(q') \le \log\delta - \min_M f < \log\epsilon -\max_M f \le \log\epsilon - f(q) = h(q)\;,
	\end{equation}
	which is a contradiction.  Therefore, $\partial \widetilde M = \emptyset$  and hence $\widetilde M= M$.
\end{proof}

The case where $\rho$, and hence  $\psi, \beta$, vanish identically, corresponds to the vacuum theory which has already been solved~\cite{Dunajski:2023xrd}. Therefore, in the following we assume that $\rho >0$ on $M$. It is worth noting that the results we have derived so far have been previously derived in the context of isolated horizons if $M=S^2$~\cite{Lewandowski:2002ua}, in which case the closed 1-form $\Phi$ must be exact and  \eqref{eq_Maxwell_X_2} is the Hodge decomposition. (For general surfaces, the  result that a non-vanishing $\rho$ must be nowhere vanishing can be deduced from \cite{Dobkowski-Rylko:2018nti}, see Remark \ref{rem_alt}).

We will now follow the method used in the vacuum theory~\cite{Dunajski:2023xrd}. Thus, for any smooth positive function $\Gamma$ we define a vector field $K\in \mathfrak{X}(M)$ by  
\begin{equation}
	K_a: = \Gamma X_a + \nabla_a \Gamma\label{eq_Kdef}  \; .
\end{equation}
This definition is inspired by the fact that the axial Killing field for the extremal Kerr-Newman horizon takes this form for some choice of $\Gamma$ (as in the vacuum case)~\cite{Kunduri:2008tk}.    The equivalent form of the Maxwell equation  \eqref{eq_Maxwell_X_2} now reads
\begin{equation}
		K^\flat = -\Gamma\star \Phi + \rho^{-1}\td(\Gamma\rho)\;.\label{eq_Maxwell_K_sol}
\end{equation}
We now turn to the quasi-Einstein equation \eqref{eq_NHG} and \eqref{eq_P_2d}, which in the new variables  becomes
\begin{equation}
	R_{ab} = \frac{K_aK_b}{2\Gamma^2}  - \frac{(\nabla_a\Gamma)(\nabla_b\Gamma)}{2\Gamma^2} - \frac{1}{\Gamma}\nabla_{(a}K_{b)} + \frac{1}{\Gamma}\nabla_a\nabla_b\Gamma + \lambda g_{ab} + \rho^2g_{ab}\;. \label{eq_NHG_rho}
\end{equation}
This is the same equation as for the vacuum case, except for the last term.  Remarkably, we find that a divergence identity analogous to the vacuum case can be still established.

\begin{proposition}
\label{prop2}
Consider a quasi-Einstein-Maxwell structure on a surface as in Proposition \ref{prop1} and assume $\rho > 0$ on $M$. Then, for any strictly positive smooth function $\Gamma$, we have the identity
\begin{align}
	\frac{1}{4}|\mathcal{L}_Kg|^2 + 2 | \nabla (\Gamma \rho) |^2 &= \nabla^a\left(K^b\nabla_{(a}K_{b)}- \frac{1}{2}K_a\Delta\Gamma-\frac{1}{2}K_a\nabla_bK^b - \lambda\Gamma K_a + \Gamma \rho \nabla_a(\Gamma \rho)\right)\nonumber\\
	&\qquad +\nabla_bK^b\left(-\frac{|K|^2}{2\Gamma}+\frac{1}{2}\Delta\Gamma+\frac{1}{2}\nabla_aK^a+\frac{1}{2\Gamma}K^a\nabla_a\Gamma+\lambda\Gamma - \Gamma \rho^2 \right) \; , \label{eq_main_id}
\end{align}
where $\Delta :=\nabla^a\nabla_a$ is the Laplacian and $| \cdot |$ denotes the $g$-norm.
\end{proposition}

\begin{proof}
We follow the derivation of the divergence identity in the vacuum theory~\cite{Dunajski:2023xrd}. This proceeds by writing
\begin{equation}
\frac{1}{4}|\mathcal{L}_Kg|^2 = \nabla_{(a} K_{b)} \nabla^a K^b= \nabla^a(K^b\nabla_{(a}K_{b)}) - K^b \nabla^a \nabla_{(a} K_{b)}  \; .
\end{equation}
The second term is computed from the contracted Bianchi identity $\nabla^a(R_{ab}-\tfrac{1}{2} R g_{ab})=0$ applied to the quasi-Einstein equation  (\ref{eq_NHG_rho}), which gives an expression for $ \nabla^a \nabla_{(a} K_{b)}$. The new term $\rho^2 g_{ab}$ in (\ref{eq_NHG_rho}) is pure trace, so this does not actually contribute to the Einstein tensor $R_{ab}-\frac{1}{2}R g_{ab}$ and therefore does not change this part of the calculation.   

The derivation then proceeds as in the vacuum case.
It turns out that the quasi-Einstein equation (\ref{eq_NHG_rho}) is only used in one other place. Namely, the expression for $ \nabla^a \nabla_{(a} K_{b)}$ includes a term with three derivatives of $\Gamma$, that is $\nabla^a\nabla_{(a}K_{b)} = \nabla^a\nabla_b\nabla_a\Gamma + ...$, and commuting  the derivatives in this term using $[\nabla_a, \nabla_b]V^a= R_{ab}V^a$ and (\ref{eq_NHG_rho}) again,  introduces an extra term  $\rho^2 \nabla_b\Gamma$. Thus, after contracting with $-K^b$,  the divergence identity of~\cite{Dunajski:2023xrd} now  becomes 
\begin{align}
	\frac{1}{4}|\mathcal{L}_Kg|^2 &= \nabla^a\left(K^b\nabla_{(a}K_{b)}- \frac{1}{2}K_a\Delta\Gamma-\frac{1}{2}K_a\nabla_bK^b - \lambda\Gamma K_a\right)\nonumber\\
	&\qquad +\nabla_bK^b\left(-\frac{|K|^2}{2\Gamma}+\frac{1}{2}\Delta\Gamma+\frac{1}{2}\nabla_aK^a+\frac{1}{2\Gamma}K^a\nabla_a\Gamma+\lambda\Gamma\right) -\rho^2K^a\nabla_a\Gamma\; . \label{eq_div_Grho}
\end{align}
The last term is the only new term compared to the vacuum theory.

We  can rewrite the term $-\rho^2K^a\nabla_a\Gamma$ using the Maxwell equation. The divergence of \eqref{eq_Maxwell_K_sol}  gives
\begin{equation}
\Gamma \rho^2 \nabla_a K^a= \rho^2 K^a\nabla_a\Gamma+ \Gamma \rho \Delta (\Gamma \rho)- | \nabla(\Gamma \rho) |^2  \; ,
\end{equation}
where we have used  \eqref{eq_Maxwell_K_sol} again to eliminate $\Phi$.  Thus, we can write
\begin{equation}
-\rho^2 K^a\nabla_a\Gamma= -\Gamma \rho^2 \nabla_a K^a+ \nabla^a(\Gamma \rho \nabla_a(\Gamma \rho))- 2 | \nabla(\Gamma \rho)|^2  \label{eq_extraterm}
\end{equation}
and substituting this into \eqref{eq_div_Grho} gives the claimed identity.
\end{proof}

We are now ready to prove our main theorem.  This together with the corresponding theorem in the vacuum case establishes Theorem \ref{main-theorem}.

\begin{theorem}\label{main-theorem-2}
Let $(M, g)$ be a $2$-dimensional compact (without boundary), oriented, Riemannian manifold admitting a quasi-Einstein-Maxwell structure with a non-gradient vector field $X$. Suppose $\rho$ does not vanish identically on $M$. Then, there exists a smooth positive function $\Gamma$ such that $K=\Gamma X+ \nabla \Gamma$ is a Killing vector field. Furthermore, $[K,X]=0$, $\mathcal{L}_K\psi=0$ and $\mathcal{L}_K \beta=0$.
\end{theorem}

\begin{proof} By~\cite[ Lemma 2.2]{Dunajski:2023xrd} there exists a smooth  function $\Gamma>0$ such that $\nabla_a K^a=0$, where $K$ is defined by \eqref{eq_Kdef} (see also~\cite{Lucietti:2012sf, Dunajski:2016rtx, tod_compact_1992}).  Proposition \ref{prop1} implies that $\rho>0$ everywhere on $M$.  With this choice of $\Gamma$,  the identity \eqref{eq_main_id} in Proposition \ref{prop2}  simplifies to
\begin{align}
	\frac{1}{4}|\mathcal{L}_Kg|^2 + 2 | \nabla (\Gamma \rho) |^2 &= \nabla^a\left(K^b\nabla_{(a}K_{b)}- \frac{1}{2}K_a\Delta\Gamma - \lambda\Gamma K_a + \Gamma \rho \nabla_a(\Gamma \rho)\right) \;. 
\end{align}
By integrating over $M$ and using Stokes' theorem, this implies $\mathcal{L}_K g=0$ and $\Gamma \rho$ is a constant everywhere on $M$.  The vector field $K$ cannot vanish identically, otherwise $X$ is a gradient. Hence, $K$ is a Killing vector field of $(M,g)$.  Constancy of $\Gamma \rho$ and \eqref{eq_extraterm}  imply that $\mathcal{L}_K\Gamma=0$ and hence $[K, X]=0$.  Furthermore, constancy of $\Gamma \rho$ and invariance of $\Gamma$ also implies $\mathcal{L}_K\rho=0$. Finally,  \eqref{eq_Maxwell_K_sol} reduces to $K^\flat= -\Gamma \star \Phi$, which implies  $\iota_K \Phi=0$, and hence from the definitions \eqref{eq_rhodef}, \eqref{eq_Phidef} we deduce that $\mathcal{L}_K \psi=\mathcal{L}_K \beta=0$.
\end{proof}

\begin{remark}
It is worth noting that the inheritance property $[K, X]=0$ follows more easily than in the vacuum case, which for $\lambda>0$ requires extra arguments~\cite{Dunajski:2023xrd, Colling:2024usk}.
\end{remark}

We conclude by demonstrating that the existence of a Killing field of the form just established has important topological restrictions.
\begin{cor}
\label{cor1}
A quasi-Einstein-Maxwell structure as in Theorem \ref{main-theorem} must be on  a sphere $M=S^2$. 
\end{cor}

\begin{proof}
The crucial observation is that $X$ must have a zero, a fact that has been previously established  in the vacuum case \cite[Lemma 3.2]{Bahuaud:2023jyq} (see also ~\cite{Jezierski:2009zz,Jezierski:2012lzy, Chrusciel:2017vie}).  This result remains valid in Einstein-Maxwell theory. Indeed, the identity $R_{ab} = \frac 12 R g_{ab}$ contracted twice with $X$ gives
\begin{equation} \label{divphi}
    X^a\nabla_a(\vert X \vert^2) = \vert X \vert^2\nabla_a X^a + \frac 12\vert X \vert^4.
\end{equation}
Following \cite{Jezierski:2009zz}, on any open set where $X$ does not vanish (\ref{divphi}) can be written as
\begin{equation}
    \nabla_a \left( \frac{X^a}{|X|^2} \right) = -\frac{1}{2} .  \label{eq_div_X}
\end{equation}
If $X$ is nowhere vanishing the vector field $X/|X|^2$ is globally defined and  \eqref{eq_div_X} is impossible on a compact manifold $M$ (without a boundary). Hence, $X$ must have a zero.

By Theorem \ref{main-theorem} there exists a function $\Gamma>0$ such that $K = \Gamma X +\nabla \Gamma$ is a Killing vector satisfying $K^a\nabla_a \Gamma = 0$. Then, at a point where $X$ vanishes, $K$ is orthogonal to itself and hence also vanishes. Since any zero of a Killing vector field on a surface must be isolated with index 1, the Poincar\'e-Hopf theorem implies that the Euler characteristic $\chi(M)$ is positive. Therefore, $M$ has to be $S^2$.
\end{proof}

\begin{remark}  
If $(M,g)$ is non-orientable, one can pull back a quasi-Einstein-Maxwell structure on $(M,g)$ to its oriented double-cover where Theorem \ref{main-theorem} applies.  Then, the push-forward of the Killing vector field $K$ on the covering space gives a Killing vector field on $(M,g)$ that leaves its quasi-Einstein-Maxwell structure invariant. In this case Corollary \ref{cor1} says that $M$ is $S^2$ or $\mathbb{R}P^2$. In fact, the extremal Kerr-Newman horizon on $S^2$ with vanishing magnetic charge descends to a solution on $\mathbb{R}P^2$.
\end{remark}

\begin{remark}\label{rem_alt}
Previously, it has been shown that in the vacuum case there are no non-trivial solutions to the quasi-Einstein equation (\ref{eq_NHG}) on higher genus surfaces~\cite{Dobkowski-Rylko:2018nti}. The proof of  Corollary \ref{cor1} also works in the vacuum case and therefore  gives an alternate proof of this result.   Conversely, one can adapt the method of~\cite{Dobkowski-Rylko:2018nti} to prove our corollary without invoking the existence of the Killing field in Theorem \ref{main-theorem}. 

To see this, consider $M$ as a Riemann surface using the complex structure induced by $g$ and set $F = \psi+ i\beta$.  Then, the Maxwell equation \eqref{eq_NH_Maxwell_2} written in any complex coordinate $z$ becomes  
\begin{equation}
    \partial_{\bar{z}} F - X_{\bar{z}} F = 0.
\end{equation}
By \cite[Lemma 1]{Dobkowski-Rylko:2018nti}, or Proposition \ref{prop1},  the function $F$ either vanishes identically or is nowhere vanishing. Therefore, if we assume the Maxwell field is non-trivial, $F$ is nowhere vanishing and one can show that $V = F^{-\frac 12}X^{(1,0)}$ is a globally defined holomorphic vector field by passing to a cover of $M$ if necessary (here $(1,0)$ refers to the decomposition $T M \otimes \mathbb{C} = T^{1,0} M \oplus T^{0,1} M$). Then, using existence properties of holomorphic vector fields on Riemann surfaces as in~\cite{Dobkowski-Rylko:2018nti}, or using that $X$ and hence $V$ must have a zero together with the fact that the zeros of holomorphic vector fields must be isolated with positive index, we deduce that $M=S^2$.
\end{remark}

\noindent{\bf Data availability.} Data sharing is not applicable to this article as no datasets were generated or analysed during the current study.\\

\noindent{\bf Competing interests.} The authors have no relevant financial or non-financial interests to disclose.

\printbibliography

\end{document}